\documentclass[journal]{IEEEtran}
\usepackage{times,amsmath,color,amssymb,graphicx,epsfig,cite,psfrag,subfigure,algorithm,balance}
\usepackage{amsfonts,pifont,enumerate,cases}
\usepackage{mathrsfs} 
\usepackage[table]{xcolor} 
\usepackage{verbatim} 
\usepackage{bm}

\newtheorem{theorem}{\underline{Theorem}}

\newtheorem{remark}{\underline{Remark}}
\newtheorem{proof}{Proof}

\def\E{\mathop\mathrm{E}}

\newcounter{MYtempeqncnt}
\begin{document}
\title{Time Varying Channel Tracking for Multi-UAV Wideband Communications with Beam Squint}
\author{Jianwei Zhao, Qi Dong, Yanjie Zhao, Bolei Wang, and Feifei Gao
	\thanks{J. Zhao, B. Wang, and F. Gao are with  Tsinghua National Laboratory for Information Science and Technology (TNList) Beijing 100084, P. R. China (e-mail: zhaojw15@mails.tsinghua.edu.cn, boleiwang@ieee.org, feifeigao@ieee.org). J. Zhao is also with
		High-Tech Institute of Xi'an, Xi'an, Shaanxi 710025, China.
Q. Dong, and Y. Zhao are with China Academy of Electronics and Information Technology, Beijing 100041, China (e-mail: dongqiouc@126.com, zhaoyj\_dky@163.com).
}
}
\maketitle
\thispagestyle{empty}
\begin{abstract}
Unmanned aerial vehicle~(UAV) has become an appealing solution for a wide range of commercial and civilian applications because of its high mobility and flexible deployment. Due to the continuous UAV navigation, the channel between UAV and base station~(BS) is subject to the Doppler effect. Meanwhile, when the BS is equipped with massive number of antennas, the non-negligible propagation delay across the array aperture would cause beam squint effect. In this paper, we first investigate the channel of UAV communications under both Doppler shift effect and beam squint effect. Then, we design a gridless compressed sensing~(GCS) based channel tracking method, where the high dimension uplink channel can be derived by estimating a few physical parameters such as the direction of arrival (DOA), Doppler shift, and the complex gain information. Besides, with the Doppler shift reciprocity and angular reciprocity, the downlink channel can be derived by only one pilot symbol, which greatly decreases the downlink channel training overhead.  Various simulation results are provided to verify the effectiveness of the proposed methods.
\end{abstract}

\section{Introduction}
Unmanned aerial vehicle~(UAV) has attracted ever increasing attention from both the industry and the academia due to its high mobility and flexible deployment. UAVs have been widely exploited in many applications such as the transportation of good, border surveillance, search and rescue missions as well as disaster response, etc~\cite{uav1,fandian}. Various UAV applications put forward exceptionally stringent communication requirements along the lines of available data rate, connection reliability, and latency,  which promotes UAV communications with massive array antennas under millimeter-wave~(mmWave) band~(30GHz-300GHz) to enhance the system performance. Different from the traditional low frequency bands~($<$ 6GHz), the mmWave has large available frequency resources that could be directly transmitted into the system bandwidth and realize broadband communications. Meanwhile, large antenna array is capable of  providing enormous spatial gain, which can be utilized to overcome the large path loss of mmWave band~\cite{mmwave2,xieCR,fandian1}.

Different from the conventional cellular communications, the majority of the UAV channel power would be contained within the line of sight~(LoS) path, which motivates a lot of angle domain signal processing studies. The authors in~\cite{zhangshun} formulated the dynamic massive MIMO channel as one sparse signal model and developed an expectation maximization (EM) based sparse Bayesian learning (SBL) framework to learn the model parameters of the sparse channel. An angle division multiple access (ADMA) based channel tracking method was proposed in~\cite{zhao2} for massive MIMO systems, where tracking the channel  is simplified to tracking the direction of arrival (DOAs) of the incident signals. Meanwhile, the uplink cooperative NOMA  was investigated in~\cite{lin} for cellular-connected UAV, which exploits the existing backhaul links among base stations to improve the throughput gains. The authors in~\cite{zhangshun2} proposed interference alignment and soft-space-reuse based cooperative transmission for multi-cell massive MIMO networks. The authors in~\cite{zhao} proposed an energy-efficient UAV communication strategy via optimizing the UAV's trajectory.

However, the channel of UAV communications with massive array antenna exhibits several unique features  compared to the conventional MIMO system, which hardens the procedure of channel tracking: (i) a practical channel of UAV communications would encounter Doppler shift due to the continuous UAV navigation; (ii) as with massive MIMO configuration, there would be a non-negligible propagation delay across the array aperture for the same data symbol, causing beam squint effect in frequency domain. Recently, there do exist some works~\cite{wang} considering the static massive MIMO with beam squint. However, channel tracking  under both Doppler shift effect and beam squint effect has not been investigated, for UAV communication systems to the best of the authors' knowledge.

In this paper, we first model the UAV communications under both Doppler shift effect and beam squint effect. Then, we present an efficient gridless compressed sensing~(GCS) based channel tracking method, where tracking the spatial channel is converted to tracking the DOA of the incident signal, Doppler shift, and complex gain respectively. Additionally, the downlink channel can be derived by only one pilot symbol with the Doppler shift and angular reciprocity. Various simulation results are provided to verify the effectiveness of the proposed studies.
\section{System and Channel Model}\label{sec:model}
We consider multiple UAV communications with  mmWave massive MIMO, where the ground base station~(BS) is equipped with the $M^B\times 1$ uniform linear array~(ULA), and $K$ UAVs are separately equipped with $M^U\times 1$ ULA, as shown in Fig.~\ref{fig:ULA}. Due to the scarce scatters in the sky, the channel is naturally sparse and there are only a few incident paths on both UAV and BS side. The large path loss of mmWave band further strengthens the channel sparsity, such that the line of sight~(LOS) path is dominant, while the other non-line of sight~(NLOS) paths can be ignored~\cite{LOS1,LOS2}. Meanwhile, due to the continuing UAV movements, the emitted wave goes through the Doppler effect.

\begin{figure}[t]
	\centering
	\includegraphics[width=85mm]{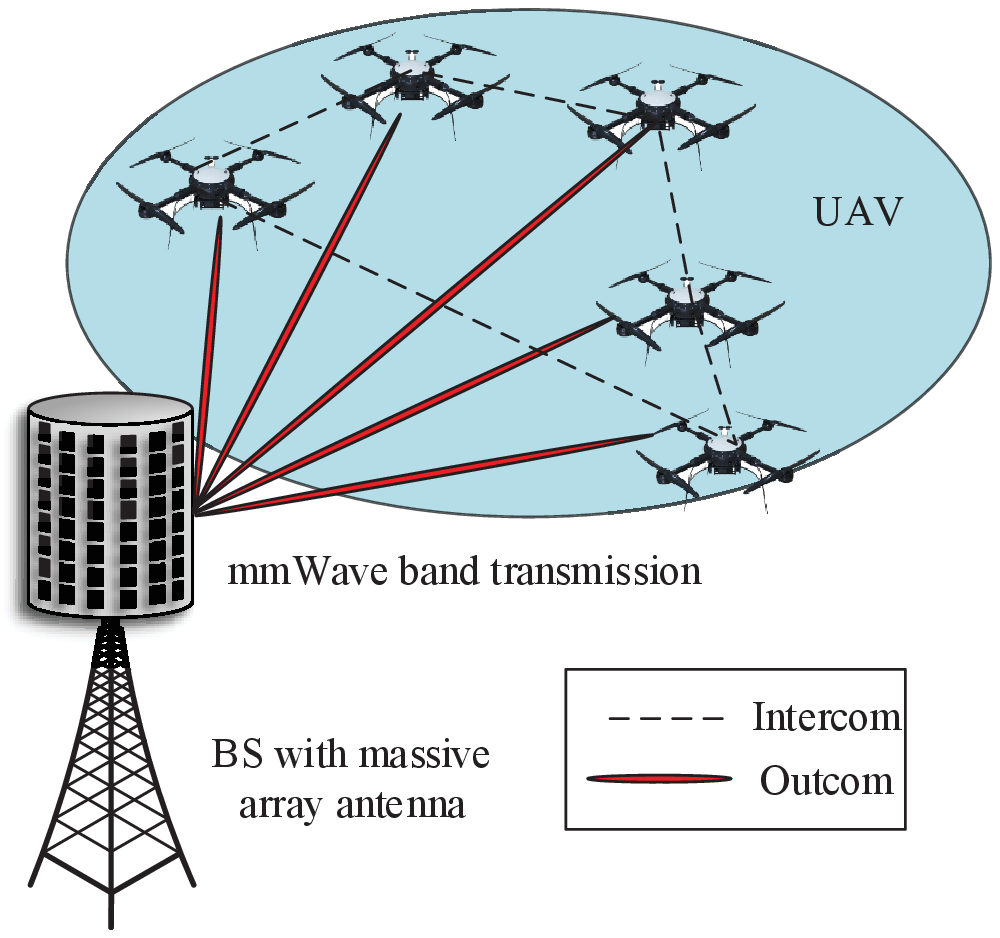}
	\caption{System Model.}
	\label{fig:ULA}
\end{figure}

As is shown in Fig.~\ref{fig:ULA1}, when the $(n+1)$-th antenna of UAV $k$ transmits the signal $s_k(t)$, the baseband signal received by the $(m+1)$-th antenna of BS from the $(n+1)$-th antenna of UAV $k$ can be denoted as~\cite{liye1,liye2}
\begin{align}\label{equ:baseband}
y_{knm}(t)=&s_k\left(t-m\frac{d\sin \theta_k^B}{c}-n\frac{d\sin \theta_k^U}{c}\right)\alpha_k \notag\\
&e^{-j2\pi f_{kd} t} e^{-j2\pi m\frac{d\sin\theta_k^B f_c }{c}}e^{-j2\pi n\frac{d\sin \theta_k^U f_c}{c}},
\end{align}
where $\alpha_k$ is the channel gain, $f_{kd}$ is the Doppler shift, $d$ is the distance between two adjacent antennas, $c$ is the light speed, $\theta_k^B$ and $\theta_k^U$ are the DOA at BS and the direction of departure~(DOD) of UAV respectively, and $f_c$ is the carrier frequency.

For the traditional narrow band MIMO systems, the antenna numbers $M^B$ and $M^U$ are finite, and meanwhile the symbol duration $T_s$ is relatively large. Hence, the following inequality always holds that
\begin{align}
m\frac{d\sin \theta_k^B}{c}+n\frac{d\sin \theta_k^U}{c} \ll T_s
\end{align}
and Equ.~\eqref{equ:baseband} reduces to
\begin{align}\label{equ:baseband1}
\!\!y_{knm}(t)\!\approx\!\!s\left(t\right) \alpha_k e^{-j2\pi f_{kd} t}\! e^{-j2\pi m\frac{d\sin \theta_k^B f_c }{c}}\!\!e^{-j2\pi n\frac{d\sin \theta_k^U f_c}{c}}.
\end{align}
In this case, the effective uplink channel between the $(n+1)$-th antenna of UAV $k$ and the $(m+1)$-th antenna of the ground BS can be expressed
\begin{align}\label{equ:ST1}
h_{k,n,m}=\alpha_k e^{-j2\pi f_{kd} t}e^{-j2\pi m\frac{d\sin \theta_k^B f_c }{c}}e^{-j2\pi n\frac{d\sin \theta_k^U f_c}{c}},
\end{align}
and the corresponding channel matrix can be derived as
\begin{align}\label{equ:nb}
\mathbf{H}_k=\alpha_{k}e^{-j2\pi f_{kd} t} \mathbf a \left(\theta_k^B\right)\mathbf {a}^T \left(\theta_k^U\right),
\end{align}
where $\mathbf a \left( \theta_k^B\right)$ is the $M^B\times 1$ steering vector at BS side with $\left[\mathbf a\left(\theta_k^B\right)\right]_m=e^{-j\frac{2\pi md\sin\theta^B}{\lambda_c}}$, while $\mathbf a \left( \theta_k^U\right)$ is the $M^U\times 1$ steering vector at UAV side with $\left[\mathbf a\left(\theta_k^U\right)\right]_n=e^{-j\frac{2\pi nd\sin\theta^U}{\lambda_c}}$.

\begin{figure}[t]
	\centering
	\includegraphics[width=85mm]{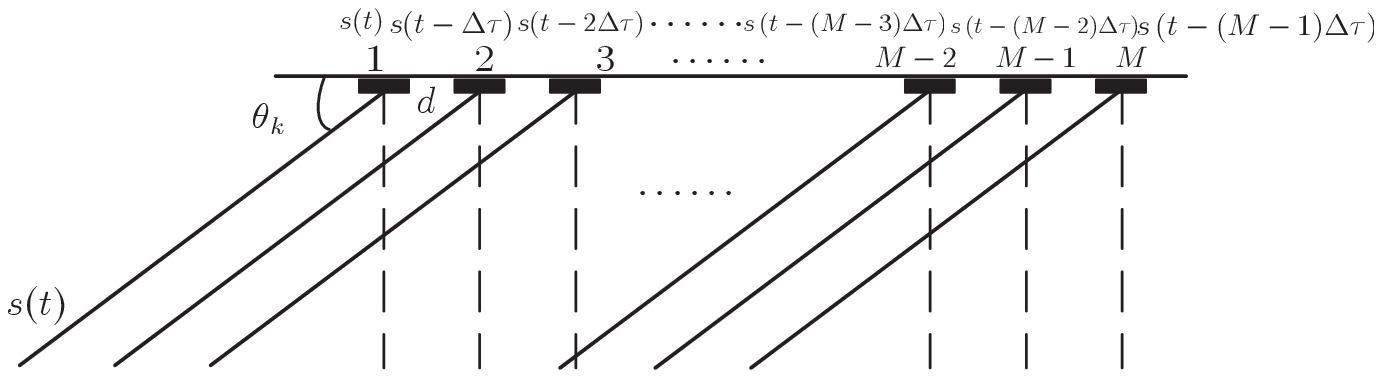}
	\caption{System Model.}
	\label{fig:ULA1}
\end{figure}

However, under massive MIMO configuration and large bandwidth of mmWave, the time delay of the signals across the large antenna array cannot be ignored. Hence, there could be 
\begin{align}
m\frac{d\sin \theta_k^B}{c}+n\frac{d\sin \theta_k^U}{c}> T_s,
\end{align}
and the approximation in Equ.~\eqref{equ:baseband1} does not hold. In this case,
different antenna would see unsynchronized $s_k(t)$ and the conventional MIMO model \eqref{equ:nb} is not valid anymore. The corresponding phenomenon can be named as \emph{beam squint} effect~\cite{wang}.


Under both the Doppler shift effect and beam squint effect, the uplink channel between the $(m+1)$-th antenna of the ground BS and the $(n+1)$-th antenna of UAV should be modeled as~\cite{liye1,liye2}
\begin{align}\label{equ:ST2}
h_{k,n,m}=&\alpha_k e^{-j2\pi f_{kd} t}e^{-j2\pi m\frac{d\sin \theta_k^B f_c }{c}}e^{-j2\pi n\frac{d\sin \theta_k^U f_c}{c}} \notag \\
&\delta \left(\tau-m\frac{d\sin \theta_k^B}{c}-n\frac{d\sin \theta_k^U}{c}\right),
\end{align}
while the corresponding frequency response can be derived as
\begin{align}\label{equ:SP2}
h_{k,n,m}(f)=&\int^{+\infty}_{-\infty}h_{k,l,m}e^{-j2\pi f\tau}d\tau \notag\\
=&\alpha_k e^{-j2\pi f_{kd} t}e^{-j2\pi m\frac{d\sin \theta_k^B f_c }{c}}e^{-j2\pi n\frac{d\sin \theta_k^U f_c}{c}} \notag \\
& e^{-j2\pi f m\frac{d\sin \theta_k^B }{c}}e^{-j2\pi f n\frac{d\sin \theta_k^U }{c}}.
\end{align}

It can be readily derived that the continuous time-frequency MIMO channel is
\begin{align}\label{equ:SP3}
\mathbf{H}_k(t,f)=\alpha_k e^{-j2\pi f_{kd} t} \mathbf a \left(\theta_k^B,f\right){\mathbf {a}}^T \left(\theta_k^U,f\right),
\end{align}
where $\mathbf a \left( \theta_k^B,f\right)$ is the $M^B\times 1$ spatial steering vector at BS side with
\begin{align}
\left[\mathbf a\left(\theta_k^B,f\right)\right]_m=e^{-j\frac{2\pi md\sin\theta_k^B}{\lambda_c}\left(1+\frac{f}{f_c}\right)},
\end{align}
while $\mathbf a \left( \theta_k^U,f\right)$ is the $M^U\times 1$ spatial steering vector at UAV side with
\begin{align}
\left[\mathbf {a}\left(\theta_k^U,f\right)\right]_n=e^{-j\frac{2\pi nd\sin\theta_k^U}{\lambda_c}\left(1+\frac{f}{f_c}\right)}.
\end{align}
\begin{remark}
To the best of the authors knowledge, this is the first work that presents  the channel modeling of massive MIMO system under both the Doppler shift effect and beam squint effect for UAV communications.
\end{remark}

For ease of illustration, we assume that each UAV is equipped with $M^U=1$ antenna in the rest of this paper. Then the continuous-time channel could be simplified as
\begin{align}\label{equ:ss}
\mathbf{H}_k(t,f)=\alpha_k e^{-j2\pi f_{d} t} \mathbf a \left(\theta_k,f\right),
\end{align}
where $\theta$ is DOA of the incident signals at BS side.

The discrete-time channel at block $l$ could be derived as
\begin{align}\label{equ:SP4}
\mathbf{h}_k(l,f)=\alpha e^{-j2\pi f_{kd} lN_bT_s} \mathbf a \left( \theta_k,f\right),
\end{align}
where $N_b$ is the number of symbols in each block.

Since the UAV speed and physical location would change much slower than the channel variation, the UAV movement related parameters such as $f_{kd}$ and $\theta_k$ can be viewed as unchanged within tens of blocks. 
Therefore, we can stack channels of $L$ blocks into an $L M\times 1$ vector $\mathbf h_k(f)$ and obtain
\begin{align}\label{equ:SP6}
\mathbf h_k(f)&=\alpha_k\textup{vec}\left[\mathbf a\left(\theta_k,f\right)\mathbf b^H\left(f_{kd}\right)\right]\notag \\ &=\alpha_k\mathbf p\left(f_{kd},\theta_k,f\right),
\end{align}
where $\mathbf b\left(f_{kd}\right)=\left[1,e^{-j2\pi f_{kd} N_bT_s},\dots, e^{-j2\pi f_{kd} (L-1)N_bT_s}\right]$ can be deemed as the Doppler steering vector.

Interestingly, even for LoS scenario, the large array would still lead to the inter symbol interference~(ISI) due to the propagation delays of the symbols across the large antenna array, which is a significantly different phenomenon from the conventional case.

Let us apply orthogonal frequency division multiplexing~(OFDM) to remove ISI. Denote $\eta=\frac{W}{N}$ as the carrier interval, where $W$ is the system bandwidth and $N$ is the number of the carriers. According to Equ.~\eqref{equ:SP6}, the channel at the $(p+1)$-th carrier can be expressed as
\begin{align}\label{equ:hk}
\mathbf h_k((p+1)\eta)=\alpha\mathbf p\left(f_{kd},\theta_k,p\eta\right).
\end{align}

Since the channel model is significant for the subsequent channel tracking, precoding, and transmission, we here provide an explicit classification rule to determine the channel model type for UAV communications: nonselective, time selective, antenna selective or doubly selective.\footnote{Doubly selectivity here means antenna selectivity plus time selectivity.} According to \eqref{equ:SP3}, the antenna selective effect would not exist, when the total time delay of the signal across the massive array meets $\frac{(M-1)d\sin \theta_k}{c}\ll T_s$, namely, $\frac{(M-1)d\sin \theta_k}{c T_s}\ll 1$. Therefore, the classification rule can be readily derived as
\begin{align}\label{equ:SP}
\max\frac{(M-1)d\sin \theta_k}{c T_s}=\frac{(M-1)d}{c T_s}\ll 1.
\end{align}

Meanwhile, there would exist time selective effect when $f_{d \max}T_s\ll 1$, and otherwise not. Therefore, the integrated channel classification rule can be seen in Tab.~\ref{table:classification}.

\begin{table}[t]
\centering
\caption{Channel classification rules.}\label{table:classification}
\rowcolors{1}{white}{gray!25}
\begin{tabular}{|c|c|c|c|c|c|c|c|c|c|c|c|c|c|c|c|c|c|}
\hline
 Classification  Rule &  $\frac{Md}{cT_s}$   &   $f_{d \max}T_s$  \\
\hline
nonselective &  $\ll$ 1    &   $\ll$ 1    \\
\hline
antenna selective   &   $\geq$ 1    &  $\ll$ 1  \\
\hline
time selective     &  $\ll$ 1     &   $\geq$ 1 \\
  \hline
doubly selective &   $\geq$ 1  &   $\geq$ 1\\
 \hline
\end{tabular}
\end{table}

\begin{theorem}
Under the condition $M^B\rightarrow \infty$ and $L\rightarrow \infty$, the channels $\mathbf h_k((p+1)\eta)$ in Equ.~\eqref{equ:SP6} are progressively orthogonal for UAV communications with mmWave massive array antenna when UAVs have distinct DOAs or velocities.
\end{theorem}

\begin{proof}
When $M^B\rightarrow \infty$ and $L\rightarrow \infty$, the relationship between $\mathbf{h}^H_1((p+1)\eta)$ and $\mathbf{h}_2((p+1)\eta)$ meets \eqref{equ:orthogonal8}, which is shown on the top of the next page. Moreover, $\xi_{k}$ is given by
\begin{align}
\xi_{k}=\left[\frac{d\left(\sin \theta_1 -\sin \theta_2\right)}{\lambda_c}+ p\eta \frac{d\left(\sin \theta_1 -\sin \theta_2\right) }{c}\right].
\end{align}

 \begin{figure*}[!t]
	\normalsize
	\setcounter{MYtempeqncnt}{\value{equation}}
	\setcounter{equation}{16}
	\begin{align}\label{equ:orthogonal8}
		&\lim_{M^B\rightarrow \infty,L\rightarrow \infty}\frac{1}{M^BL}\mathbf{h}^H_1((p+1)\eta)\mathbf{h}_2((p+1)\eta)=\lim_{M^B\rightarrow \infty}\frac{1}{M^BL}\alpha_{1}^{\star}\alpha_{2}\mathbf p^H(\theta_1)\mathbf p(\theta_2)\notag\\
		&=\lim_{M^B\rightarrow \infty,L\rightarrow \infty}\frac{1}{M^BL}\alpha_{1}^{\star}\alpha_{2}\left\{\textup{vec}\left[\mathbf a\left(\theta_{1},(p+1)\eta\right)\mathbf b^H\left(f_{1d}\right)\right]\right\}^H\textup{vec}\left[\mathbf a\left(\theta_{2},(p+1)\eta\right)\mathbf b^H\left(f_{2d}\right)\right]\notag\\
		&=\lim_{M^B\rightarrow \infty,L\rightarrow \infty}\frac{1}{M^BL}\sum_{l=1}^{L}\alpha_{1}^{\star}\alpha_{2}e^{-j2\pi (f_{1d}-f_{2d})  (l-1)N_bT_s}\mathbf a^H(\theta_1,(p+1)\eta)\mathbf a(\theta_2,(p+1)\eta)\notag\\
		&=\lim_{L\rightarrow\infty}\frac{\alpha_{1}^{\star}\alpha_{2}e^{-j\pi(M^B-1)\xi_{k}}e^{-j\pi(L-1)(f_{1d}-f_{2d})N_bT_s}}{L}\frac{\sin\left(\pi L(f_{1d}-f_{2d})N_bT_s \right)}{\sin\left(\pi(f_{1d}-f_{2d})N_bT_s \right)}\delta(\theta_1-\theta_2)\notag\\
		&=\alpha_{1}^{\star}\alpha_{2}e^{-j\pi (M^B-1)\xi_{k}}e^{-j\pi(L-1)(f_{1d}-f_{2d})N_bT_s}\delta(f_{1d}-f_{2d})\delta(\theta_1-\theta_2).
	\end{align}
	\setcounter{equation}{\value{MYtempeqncnt}}
	\addtocounter{equation}{1}
	\hrulefill
	\vspace*{4pt}
\end{figure*}

Therefore, the channels $\mathbf h_k((p+1)\eta)$ in Equ.~\eqref{equ:SP6} are progressively orthogonal for UAV communications  when the UAVs have distinct DOAs or velocities.

\end{proof}
\begin{figure}[t]
	\centering
	\includegraphics[width=75mm]{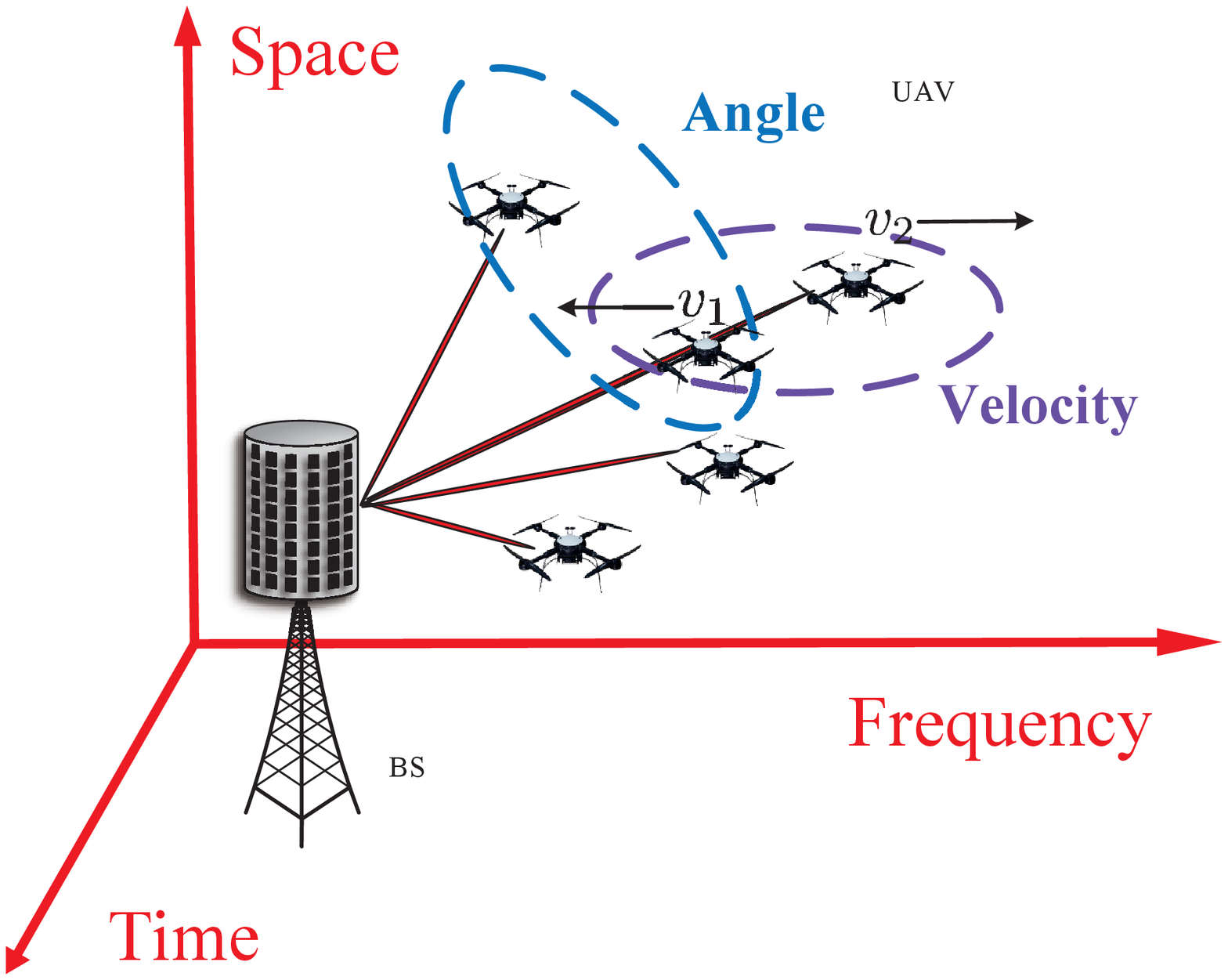}
	\caption{User scheduling scheme, where users with different DOAs or velocities could be simultaneously scheduled.}
	\label{fig:pdma1}
\end{figure}

According to Theorem 1, the sparse characteristic of UAV communications with mmWave massive array antenna makes it possible to schedule users according to the channel DOA or velocity information. As is vividly shown in Fig.~\ref{fig:pdma1}, the users with different DOA or velocity could be simultaneously scheduled. The corresponding user scheduling scheme can be named as angle division multiple access~(ADMA) and velocity division multiple access~(VDMA).
\section{Channel Tracking Strategy}\label{sec:channeltracking}
The beam squint effect makes the traditional channel transmission strategy inapplicable for UAV communications with mmWave massive array antenna. In this section, we will provide a GCS based channel tracking method for UAV communications with mmWave massive array antenna.
\subsection{Uplink Channel Tracking}
We here utilize the comb-type pilot channel estimation to track the time-varying channel. Let us assume $P$ of $N_c$ subcarriers are exploited as pilots, and the corresponding  subcarrier index set for user $k$ is $\mathcal{P}_k=\left\{p_{k,1},\dots,p_{k,P}\right\}$. Then, the channel of user $k$ in these pilot subcarrier can be stacked into a matrix as
\begin{align}\label{equ:SPs}
\mathbf H_k&=\alpha_{k}\left[\mathbf p\left(f_{kd},\theta_k,0\right),\dots,\mathbf p\left(f_{kd},\theta_k,(P-1)\eta\right)\right]\notag\\&= \mathbf{P}\left(\theta_k,f_{kd}\right)\alpha_k.
\end{align}

Assuming that all users send pilot symbol ``1'' over the selected sub-carriers while transmitting data symbols over other sub-carriers. Then, the received uplink pilots from the $M$ antennas and $P$ subcarriers over $L$ blocks can be derived as
\begin{align}
\mathbf Y=\sum_{k=1}^{K}\mathbf{H}_k\mathbf X+\mathbf W,
\end{align}
where $\mathbf X\in\mathbb{C}^{P\times P}$ is the diagonal pilot matrix whose elements are in the selected subcarriers, $K$ is the number of the scheduled UAVs, and $\mathbf W$ is the additive Gaussian noise whose elements are independently distributed as $\mathcal{CN}(0,\sigma^2)$.

Denote $\mathbf y=\textup{vec}(\mathbf Y\mathbf X^{-1})$, $\mathbf h_k=\textup{vec}(\mathbf H_k)$ and $\mathbf w=\textup{vec}(\mathbf W\mathbf X^{-1})$. Then, we have
\begin{align}
\mathbf y=\sum_{k=1}^{K}\mathbf{h}_k+\mathbf w=\boldsymbol{\alpha}\mathbf P(\boldsymbol\theta,\boldsymbol f_d)+\mathbf w.
\end{align}

According to the channel model \eqref{equ:hk}, the channel can be determined by the complex  channel gain parameter $\alpha_k$, the Doppler shifter parameter $f_{kd}$, and the DOA $\theta_k$. Therefore, the high dimension channel tracking problem can be transformed into estimating a few dominant channel physical parameters $\alpha_k$, $f_{kd}$, and $\theta_k$.

Since the number of the channel parameters is far less than the dimension of the channel, the compressive sensing~(CS) become an effective approach to estimate the unknown channel parameters.  Here, we adopt the GCS based parameter estimation method to track the channel since it can provide the off-grid angular parameter estimation. The initial number of the scheduled UAVs is set as $K_M$, where $K_M\geq K$ to guarantee enough degree of freedom for the UAV number estimation. Then, the problem can be formulated as
\begin{align}\label{equ:Op}
&\min_{\boldsymbol\theta,\boldsymbol f_d,\boldsymbol\alpha}\ \ \ \parallel\boldsymbol\alpha\parallel_0\notag\\
&s.t.\ \ \ \parallel\mathbf y-\boldsymbol{\alpha}\mathbf P(\boldsymbol\theta,\boldsymbol f_d)\parallel_2\leq \xi,
\end{align}
where $\parallel\boldsymbol\alpha\parallel_0$ represents the number of the nonzero entries of $\boldsymbol\alpha$, and $\xi$ is a small positive number that controls the error tolerability of the noise statistics.

The optimization in Equ.~\eqref{equ:Op} is an NP-hard problem, and the log-sum sparsity encouraging function, i.e.,
\begin{align}
J_0(\alpha)\triangleq \sum_{k=1}^{K}\log \left(\left|[\boldsymbol\alpha]_k\right|^2+\epsilon\right),
\end{align}
can be exploited to derive the one to  the equivalent optimization objective
\begin{align}\label{equ:Op2}
&\min_{\boldsymbol\theta,\boldsymbol f_d,\boldsymbol\alpha}\ \ \ J_0(\alpha)\notag\\
&s.t.\ \ \ \parallel\mathbf y-\boldsymbol{\alpha}\mathbf P(\boldsymbol\theta,\boldsymbol f_d)\parallel_2\leq \xi,
\end{align}
where $\epsilon$ is a iterative parameter.

Next, we introduce the data fitting term $\lambda \parallel\mathbf y-\boldsymbol{\alpha}\mathbf P(\boldsymbol\theta,\boldsymbol f_d)\parallel_2^2$ to eliminate the constraint $\parallel\mathbf y-\boldsymbol{\alpha}\mathbf P(\boldsymbol\theta,\boldsymbol f_d)\parallel_2\leq \xi$. In this way, the optimization problem is converted into
\begin{align}
\min_{\boldsymbol\theta,\boldsymbol f_d,\boldsymbol\alpha}\ J_{\lambda}(\boldsymbol\theta,\boldsymbol f_d,\boldsymbol\alpha)=\sum_{k=1}^{K}&\log \left(\left|[\boldsymbol\alpha]_k\right|^2+\epsilon\right)\notag\\
&+\lambda {\parallel\mathbf y-\boldsymbol{\alpha}\mathbf P(\boldsymbol\theta,\boldsymbol f_d)\parallel}^2_2,
\end{align}
where the parameter $\lambda$ determines the compromise between the sparsity and the data fitting deviation.
The larger $\lambda$ puts more weight on the fitting deviation, and therefore produces a better-fitting solution, but would also increase the possibility of overestimation. On the contrary, a smaller one will make the optimization converge to a sparser result and an underestimated solution.
Here, $\lambda$ is set as the inverse of the noise variance of vector $\boldsymbol{\alpha}$'s elements to achieve the tradeoff between the sparsity and the data-fitting deviation, which is given by
\begin{align}\label{equ:lambda}
\lambda=\max\left( \lambda_0\frac{1}{\parallel\mathbf y-\boldsymbol{\alpha}\mathbf P(\boldsymbol\theta,\boldsymbol f_d)\parallel_2^2},\lambda_{\min}\right),
\end{align}
where $\lambda_{0}$ and $\lambda_{\min}$ are two constants. Moreover, during the optimization progress, $\lambda$ will be dynamically adjusted until it reaches  $\lambda_{\min}$, and $\lambda_{0}$ remains fixed to balance between the sparsity and the data-fitting deviation.

By exploiting the maximization-minimization~(MM) iterative method \cite{mm1}, the surrogate function $S_0\left(\boldsymbol{\alpha}|\boldsymbol{\alpha}^{(n)}\right)$ can be derived to minimize in the iterations of maximizing $J_0(\boldsymbol\alpha)$, which is given by
\begin{align}\label{equ:Op5} &\sum_{k=1}^{K}\left[\left(\left|[\boldsymbol\alpha^{(n)}]_k\right|^2+\epsilon\right)+\frac{\left(\left|[\boldsymbol\alpha]_k\right|^2+\epsilon\right)-\left(\left|[\boldsymbol\alpha^{(n)}]_k\right|^2+\epsilon\right)}{\left|[\boldsymbol\alpha^{(n)}]_k\right|^2+\epsilon}\right]\notag \\ &\geq J_0(\boldsymbol{\alpha}),
\end{align}
where $\boldsymbol \alpha^{(n)}$ is the estimated value of the complex gain at the $n$-th iteration. 

The last inequality in~\eqref{equ:Op5} results from the convexity of $-J_0(\boldsymbol{\alpha})$, and the equality will be attained only when $\boldsymbol{\alpha}^{(n)}=\boldsymbol{\alpha}$. Therefore, at the $(n+1)$-th iteration, it will hold that
\begin{align}\label{equ:Op6}
S_0\left(\boldsymbol{\alpha}|\boldsymbol{\alpha}^{(n)}\right)-J_0(\boldsymbol{\alpha})\geq S_0\left(\boldsymbol{\alpha}^{(n)}|\boldsymbol{\alpha}^{(n)}\right)-J_0(\boldsymbol{\alpha}^{(n)}).
\end{align}

Then, the optimization problem  can be transformed into
\begin{align}\label{equ:Op7}
&\min_{\boldsymbol \theta, \boldsymbol f_d, \boldsymbol \alpha}S_{\lambda}\left(\boldsymbol{\theta},\boldsymbol{f}_d,\boldsymbol{\alpha}|\boldsymbol{\alpha}^{(n)}\right)\notag\\
&=S_0\left(\boldsymbol{\alpha}|\boldsymbol{\alpha}^{(n)}\right)+\lambda \parallel\mathbf y-\boldsymbol{\alpha}\mathbf P(\boldsymbol\theta,\boldsymbol f_d)\parallel_2^2\notag\\
&=\sum_{k=1}^{K}\frac{\left|[\boldsymbol\alpha]_k\right|^2}{\left|[\boldsymbol\alpha^{(n)}]_k\right|^2+\epsilon}+\lambda \parallel\mathbf y-\boldsymbol{\alpha}\mathbf P(\boldsymbol\theta,\boldsymbol f_d)\parallel_2^2+C\left(\boldsymbol \alpha^{(n)}\right),
\end{align}
where $C\left(\boldsymbol \alpha^{(n)}\right)$ is a constant that is independent of $\boldsymbol \theta, \boldsymbol f_d, \boldsymbol \alpha$.

We denote $\mathbf{D}^{(n)}=\textup{diag}\left\{\frac{\left|[\boldsymbol\alpha]_1\right|^2}{\left|[\boldsymbol\alpha^{(n)}]_1\right|^2+\epsilon},\dots,\frac{\left|[\boldsymbol\alpha]_{K_i}\right|^2}{\left|[\boldsymbol\alpha^{(n)}]_{K_i}\right|^2+\epsilon}\right\}$, and it can be readily derived that
\begin{align}
&S_{\lambda}\left(\boldsymbol{\theta},\boldsymbol{f}_d,\boldsymbol{\alpha}|\boldsymbol{\alpha}^{(n)}\right)=\notag\\&\ \ \ \ \ \ \boldsymbol\alpha ^{H}\mathbf{D^{(n)}}\boldsymbol\alpha+\lambda \parallel\mathbf y-\boldsymbol{\alpha}\mathbf P(\boldsymbol\theta,\boldsymbol f_d)\parallel_2^2+C\left(\boldsymbol \alpha^{(n)}\right).
\end{align}

According to the above analysis, we can further compute that
\begin{align}\label{equ:Op9}
&J_{\lambda}(\boldsymbol\theta^{(n+1)},\boldsymbol f_d^{(n+1)},\boldsymbol\alpha^{(n+1)})\notag\\
&=J_0(\alpha^{(n+1)})+\lambda \parallel\mathbf y-\boldsymbol{\alpha^{(n+1)}}\mathbf P\left(\boldsymbol \theta^{(n+1)}, \boldsymbol f_d^{(n+1)}\right)\parallel_2^2\notag\\
&\leq S_0\left(\boldsymbol{\alpha}^{(n)}|\boldsymbol{\alpha}^{(n)}\right)-S_0\left(\boldsymbol{\alpha}^{(n)}|\boldsymbol{\alpha}^{(n)}\right)+J_0(\alpha^{(n)})+\notag\\
&\lambda \!\parallel\mathbf y-\boldsymbol{\alpha^{(n)}}\mathbf P\left(\boldsymbol \theta^{(n)}, \boldsymbol f_d^{(n)}\right)\parallel_2^2\notag\\
&= J_{\lambda}(\boldsymbol\theta^{(n)},\boldsymbol f_d^{(n)},\boldsymbol\alpha^{(n)}).
\end{align}

Equ.~\eqref{equ:Op9} means that decreasing the surrogate function $S_{\lambda}\left(\boldsymbol{\theta},\boldsymbol{f}_d,\boldsymbol{\alpha}|\boldsymbol{\alpha}^{(n)}\right)$ indeed decreases $J_{\lambda}(\boldsymbol\theta,\boldsymbol f_d,\boldsymbol\alpha)$, which guarantees the effectiveness of optimizing~\eqref{equ:Op7}. Therefore, we only need to minimize the surrogate function $S_{\lambda}\left(\boldsymbol{\theta},\boldsymbol{f}_d,\boldsymbol{\alpha}|\boldsymbol{\alpha}^{(n)}\right)$. Then, for the given $\boldsymbol{\theta}$ and $
\boldsymbol{f}_d$, the optimal value of $\boldsymbol{\alpha}$ can be immediately derived as
\begin{align}\label{equ:Op10}
&\boldsymbol\alpha^{\star}(\boldsymbol\theta,\boldsymbol f_d)=\notag\\
&\left(\mathbf P^H\left(\boldsymbol \theta, \boldsymbol f_d\right)\mathbf P\left(\boldsymbol \theta, \boldsymbol f_d\right)+\lambda^{(-1)}\mathbf{D}^{(n)} \right)^{-1} \mathbf P^H\left(\boldsymbol \theta, \boldsymbol f_d\right)\mathbf{y}.
\end{align}

When we substitute  $\boldsymbol\alpha^{\star}(\boldsymbol\theta,\boldsymbol f_d)$ into \eqref{equ:Op7}, the optimization are converted into
\begin{align}\label{equ:Op11}
&\min_{\boldsymbol \theta, \boldsymbol f_d}S_{1}\left(\boldsymbol{\theta},\boldsymbol{f}_d\right)=C\left(\boldsymbol \alpha^{(n)}\right)-\mathbf{y}^H\mathbf P\left(\boldsymbol \theta, \boldsymbol f_d\right)\notag\\
&\left(\mathbf P^H\left(\boldsymbol \theta, \boldsymbol f_d\right)\mathbf P\left(\boldsymbol \theta, \boldsymbol f_d\right)+\lambda^{(-1)}\mathbf{D}^{(n)} \right)^{-1} \mathbf P^H\left(\boldsymbol \theta, \boldsymbol f_d\right)\mathbf{y}.
\end{align}


Since $\min_{\boldsymbol \theta, \boldsymbol f_d}S_{1}\left(\boldsymbol{\theta},\boldsymbol{f}_d\right)$ is differentiable with respect to $\boldsymbol{\theta}$ and $\boldsymbol{f}_d$, the gradient descent can be exploited in each iteration to derive $\boldsymbol \theta$ and $\boldsymbol f_d$.

Define 
\begin{align}
&\mathbf Z=\mathbf P\left(\boldsymbol \theta, \boldsymbol f_d\right)\mathbf R \mathbf P^H\left(\boldsymbol \theta, \boldsymbol f_d\right)\mathbf{y},\\
&\mathbf R=\left(\mathbf P^H\left(\boldsymbol \theta, \boldsymbol f_d\right)\mathbf P\left(\boldsymbol \theta, \boldsymbol f_d\right)+\lambda^{(-1)}\mathbf{D}^{(n)} \right)^{-1}. 
\end{align}

Then, it can be readily derived as \eqref{equ:derivation}, which is shown on the top of the next page. 

 \begin{figure*}[!t]
	\normalsize
	\setcounter{MYtempeqncnt}{\value{equation}}
	\setcounter{equation}{35}
\begin{align}\label{equ:derivation}
\frac{\partial S_1\left(\boldsymbol \theta, \boldsymbol f_d\right)} {\partial {\theta }_{k}}=\text{tr} \left\{{\left(\frac{\partial S_1\left(\boldsymbol \theta, \boldsymbol f_d\right)} {\partial \boldsymbol{Z}}\right)}^{T}\frac{\partial \boldsymbol{Z}} {\partial {\theta }_{k}}\right\}+\text{tr} \left\{{\left(\frac{\partial S_1\left(\boldsymbol \theta, \boldsymbol f_d\right)} {\partial {\boldsymbol{Z}}^{*}}\right)}^{T}\frac{\partial {\boldsymbol{Z}}^{*}} {\partial {\theta }_{k}}\right\},
\end{align}
	\setcounter{equation}{\value{MYtempeqncnt}}
	\addtocounter{equation}{1}
	\hrulefill
	\vspace*{4pt}
\end{figure*}

since it holds that 
\begin{align}
&{\frac{\partial S_1\left(\boldsymbol \theta, \boldsymbol f_d\right)} {\partial \boldsymbol{Z}}} =\frac{\partial } {\partial \boldsymbol{Z}}\text{tr} \left\{-\boldsymbol{y} {\boldsymbol{y}}^{H}\boldsymbol{Z}\right\}=-{\left(\boldsymbol{y} {\boldsymbol{y}}^{H}\right)}^{T}\\
&\frac{\partial S_1\left(\boldsymbol \theta, \boldsymbol f_d\right)} {\partial {\boldsymbol{Z}}^{*}} =\frac{\partial } {\partial {\boldsymbol{Z}}^{*}}\text{tr} \left\{-\boldsymbol{y} {\boldsymbol{y}}^{H}\boldsymbol{Z}\right\}=\mathbf{0},
\end{align}
the gradient of $S_{1}\left(\boldsymbol{\theta},\boldsymbol{f}_d\right)$ corresponding to ${\theta_k}$ can be derived as
\begin{align}
\frac{\partial S_1\left(\boldsymbol \theta, \boldsymbol f_d\right)} {\partial {\theta }_{k}}=\text{tr} \left\{-\boldsymbol{y} {\boldsymbol{y}}^{H}\frac{\partial \boldsymbol{Z}} {\partial {\theta }_{k}} \right\},
\end{align}
where $\frac{\partial \boldsymbol{Z}} {\partial {\theta }_{k}}$ is given
by
\begin{align}
&\frac{\partial \boldsymbol{Z}} {\partial {\theta }_{k}}=\frac{\partial \mathbf P\left(\boldsymbol \theta, \boldsymbol f_d\right)} {\partial {\theta }_{k}} \mathbf R{\mathbf P}^{H} \left(\boldsymbol \theta, \boldsymbol f_d\right)+\mathbf P\left(\boldsymbol \theta, \boldsymbol f_d\right)\mathbf R\frac{\partial {\mathbf P}^{H}\left(\boldsymbol \theta, \boldsymbol f_d\right)} {\partial {\theta }_{k}}-\notag\\
&\mathbf P\left(\boldsymbol \theta, \boldsymbol f_d\right)\mathbf R \left(\frac{\partial {\mathbf P}^{H} \left(\boldsymbol \theta, \boldsymbol f_d\right)} {\partial {\theta }_{k}}\mathbf P \left(\boldsymbol \theta, \boldsymbol f_d\right)+{\mathbf P}^{H} \left(\boldsymbol \theta, \boldsymbol f_d\right)\frac{\partial \mathbf P\left(\boldsymbol \theta, \boldsymbol f_d\right)} {\partial {\theta }_{k}}\right)\notag\\&\mathbf R {\mathbf P}^{H} \left(\boldsymbol \theta, \boldsymbol f_d\right).
\end{align}
Moreover, $\frac{\partial \mathbf P\left(\boldsymbol{\theta }\right)} {\partial {\theta }_{k}}$ can be derived as
\begin{align}
\frac{\partial \mathbf P\left(\boldsymbol \theta, \boldsymbol f_d\right)} {\partial {\theta }_{k}}=\mathbf Q \odot \mathbf P\left(\boldsymbol \theta, \boldsymbol f_d\right),
\end{align}
where $\mathbf Q\in \mathbb{C}^{M^BPT\times K}$, and the $k$-th row of $\mathbf Q$ is given as \eqref{equ:q} while the elements of other rows are all zero.

\begin{figure*}[!t]
	\normalsize
	\setcounter{MYtempeqncnt}{\value{equation}}
	\setcounter{equation}{41}
	\begin{align}\label{equ:q}
\mathbf q=-j2\pi\cdot \left[ \underbrace{\textup{vec} \left\lbrace[0:M^B-1]^T*(1+\eta/f_c*\mathbf 0_{L\times 1}^T)\right\rbrace^T,\dots,\textup{vec} \left\lbrace[0:M^B-1]^T*(1+\eta/f_c*\mathbf {(P-1)}_{L\times 1}^T)\right\rbrace^T}_{P}\right]\in \mathbb{C}^{M^BPL\times 1},
	\end{align}
	\setcounter{equation}{\value{MYtempeqncnt}}
	\addtocounter{equation}{1}
	\hrulefill
	\vspace*{4pt}
\end{figure*}

The gradient of $S_{1}\left(\boldsymbol{\theta},\boldsymbol{f}_d\right)$ corresponding to ${f_{kd}}$ can be derived similarly, which is given by
\begin{align}
\frac{\partial S_1\left(\boldsymbol \theta, \boldsymbol f_d\right)} {\partial {f }_{kd}}=\text{tr} \left\{-\boldsymbol{y} {\boldsymbol{y}}^{H}\frac{\partial \boldsymbol{Z}} {\partial {f}_{kd}} \right\},
\end{align}
where $\frac{\partial \boldsymbol{Z}} {\partial {f}_{kd}}$ is given
by
\begin{align}
 &\frac{\partial \boldsymbol{Z}} {\partial {\theta }_{k}}  =\frac{\partial \mathbf P\left(\boldsymbol \theta, \boldsymbol f_d\right)} {\partial {f }_{kd}} \mathbf R{\mathbf P}^{H} \left(\boldsymbol \theta, \boldsymbol f_d\right)+\mathbf P\left(\boldsymbol \theta, \boldsymbol f_d\right)\mathbf R\frac{\partial {\mathbf P}^{H} \left(\boldsymbol \theta, \boldsymbol f_d\right)} {\partial {f }_{kd}}-\notag\\
&\mathbf P\left(\boldsymbol \theta, \boldsymbol f_d\right)\mathbf R \left(\frac{\partial {\mathbf P}^{H} \left(\boldsymbol \theta, \boldsymbol f_d\right)} {\partial {f }_{kd}}\mathbf P \left(\boldsymbol \theta, \boldsymbol f_d\right)+{\mathbf P}^{H} \left(\boldsymbol \theta, \boldsymbol f_d\right)\frac{\partial \mathbf P\left(\boldsymbol \theta, \boldsymbol f_d\right)} {\partial {f }_{kd}}\right)\notag\\
&\mathbf R {\mathbf P}^{H} \left(\boldsymbol \theta, \boldsymbol f_d\right).
\end{align}
Moreover,  $\frac{\partial \mathbf P\left(\boldsymbol \theta, \boldsymbol f_d\right)} {\partial {f }_{kd}}$ can be derived as
\begin{align}
\frac{\partial \mathbf P\left(\boldsymbol \theta, \boldsymbol f_d\right)} {\partial {f }_{kd}}=\mathbf U \odot \mathbf P\left(\boldsymbol \theta, \boldsymbol f_d\right),
\end{align}
where $\mathbf U\in \mathbb{C}^{MPT\times G_k}$, and the $k$-th row of $\mathbf U$ is given by $\left[ \underbrace{\mathbf u^T,\dots,\mathbf u^T}_{P}\right]^T $, and the elements of other rows are all zero.
Moreover, $\mathbf u$ is given by
\begin{align}
\mathbf u=\textup{vec} \left\lbrace -1j*2*\pi*\mathbf 1_{M^B\times 1}*[0:L-1]\right\rbrace \in \mathbb{C}^{M^BL\times 1}.
\end{align}

The concrete steps of the proposed algorithm are displayed in Alg.~\ref{alg:algorithm1}. Since the dictionary of the  GCS approach is not pre-defined, which remains unknown in the process of the parameter estimation, the proposed channel tracking method overcomes the performance degradation of the traditional on-grid CS methods due to the grid mismatch.

\begin{algorithm}[t]
	\caption{:Uplink channel tracking}
	\label{alg:algorithm1}
	\begin{itemize}
		\item \textbf{Step 1:} 
		Set $n=0$ and $K=K_M$; initialize $\boldsymbol{\alpha^{(0)}}$, $\boldsymbol{f}_d$, and $\boldsymbol{\theta}$; compute $\lambda^{(n)}$ according to~\eqref{equ:lambda}.		
		\item \textbf{Step 2:} At the $n$-th iteration, construct the surrogate function according to \eqref{equ:Op11};
		\item \textbf{Step 3:}By exploiting the gradient descend, the surrogate function is optimized to find a new iterative estimate of $\boldsymbol \theta$ and $\boldsymbol{f}_d$;
		\item \textbf{Step 4:} Calculate $\boldsymbol{\alpha}^{(n)}$ according to \eqref{equ:Op10}, and update $\lambda^{(n+1)}$ according to \eqref{equ:lambda};
		\item \textbf{Step 5:} Compute $\gamma=\left\|\boldsymbol{\alpha^{(n+1)}}-\boldsymbol{\alpha^{(n)}}\right\|_2 $. If $\gamma<\sqrt{\epsilon}$, then $\epsilon=\max\left\lbrace \frac{1}{\epsilon}, \epsilon_{\min}\right\rbrace $;
		\item \textbf{Step 6:} For $l$ satisfying $\left[ \boldsymbol{\alpha}^{(n+1)}\right]_l<\alpha_{\min}$, remove $\left[ \boldsymbol{\alpha}^{(n+1)}\right]_l$, $\left[ \boldsymbol{f}_d^{(n+1)}\right]_l$, and $\left[ \boldsymbol{\theta}^{(n+1)}\right]_l$ from $\boldsymbol{\alpha}^{(n+1)}$, $\boldsymbol{f}_d^{(n+1)}$, and $ \boldsymbol{\theta}^{(n+1)}$;
		\item \textbf{Step 7:} Set $n=n+1$; if $\gamma<\gamma^S$, and go to \textbf{Step 2}, where $\gamma^S$ is the hard threshold as the terminating condition. Otherwise stop the iteration, and output the results.
	\end{itemize}
\end{algorithm}
\subsection{Downlink Channel Tracking with the Angular and Doppler Shift Reciprocity}
\subsubsection{Downlink Channel Representation}
According to \cite{Bultitude,METIS}, the physical DOAs $\theta_{k}$ are approximately identical for the uplink and downlink channel transmission, namely,
\begin{align}
\theta^D_{k}=\theta_{k},
\end{align}
which is called \emph{angular reciprocity}. The angular reciprocity holds true at the case that the frequency interval between the uplink and downlink channel is within several GHz. 

Meanwhile, since the relative velocity of the uplink and downlink channel are the same, the downlink Doppler shift can also be derived from the uplink one, which can be name as \emph{Doppler shift reciprocity}. We denote the downlink channel carrier frequency and its corresponding carrier wavelength as $f_c^D$ and $\lambda_c^D$. With the uplink Doppler shift $f_{kd}$, the downlink one can be derived as
\begin{align}\label{equ:doppler}
f_{kd}^D=f_{kd}\frac{f_c^D}{f_c}.
\end{align}

Then, with both the angular and Doppler shift reciprocity, the downlink channel over all the sub-carriers can be formulated as
\begin{align}\label{equ:dcf}
\mathbf g_k(f)&=\alpha_{k}^D\textup{vec}\left[\mathbf a\left(\theta_{k}^D,f\right)\mathbf b^H\left(f_{kd}^D\right)\right]\notag\\
&=\alpha_{k}\mathbf p\left(f_{kd}^D,\theta_k^D,f\right).
\end{align}

\subsubsection{Downlink Channel Tracking} According to~\eqref{equ:dcf}, the only unknown parameter in the downlink channel $\mathbf g_k(f)$ is the complex gain $\alpha^D_{k}$. Therefore, we only need to estimate $\alpha^D_{k}$ to track the total downlink channel. 

Denote the beamforming vector for user $k$ of the $n$-th sub-carrier as 
\begin{align}
\mathbf{g}_k(n)=\left[ \mathbf a^D \left( \theta_k^D,n\right) \right] ^H,
\end{align}
and then overall beamforming matrix at BS can be expressed as
\begin{align}\label{equ:beamforming}
\mathbf{g}^D(n)=\sum_{k=1}^{K_i}\mathbf{g}_k(n),
\end{align}
where the beamforming vector $\mathbf{g}^D(n)$ says that BS formulate beams towards DOA of the scheduled UAVs. 

Denote $s$ as the training symbol, and then we select the sub-carrier $n$ for downlink channel tracking. The angle domain sparse channels permit UAVs with different DOAs to be trained by the same pilot sequence, and therefore the same one pilot symbol $s$ can be simultaneous utilized to decrease the training overhead.

 The received signal of the $p$-th UAV at the $n$-th sub-carrier can be expressed as
\begin{align}\label{equ:downlinkr}
&y_p(l,n)=\mathbf{h}^D_k(l,n)\mathbf{g}^D(n)s+\omega_k^D\notag\\
&=\mathbf{h}^D_k(l,n)\mathbf{g}_k(n)s+\sum_{k'=1,k'\neq k}^{K_i}\mathbf{h}^D_{k'}(l,n)\mathbf{g}_{k'}(n)s+\omega_k^D,
\end{align}
where $\sum_{k'=1,k'\neq k}^{K_i}\mathbf{h}^D_{k'}(l,n)\mathbf{g}_{k'}(n)s$ is the interference. Nevertheless, since all the scheduled UAVs has distinct DOA, it holds that $\sum_{k'=1,k'\neq k}^{K_i}\mathbf{h}^D_{k'}(l,n)\mathbf{g}_{k'}(n)s\approx 0$.

Then, the $p$-th UAV sums the received signals from all the sub-carrier, which is given by
\begin{align}\label{equ:downlink2}
y_p(l)&=\sum_{n=1}^{N}y_p(l,n)=N\alpha^D_{k}e^{-j2\pi f_{kd}^D lN_bT_s}s.
\end{align}

Therefore, the downlink channel complex gain can be derived as
\begin{align}\label{equ:downlinkc}
\hat\alpha_k^D=\frac{y_p(l)}{Ne^{-j2\pi f_{kd}^D lN_bT_s}s},
\end{align}
while the downlink channel can be reconstructed as
\begin{align}\label{equ:dccf}
\hat{\mathbf g}_k(f)=\hat{\alpha}_{k}^D\textup{vec}\left[\mathbf a\left(\hat{\theta}_{k}^D,f\right)\mathbf b^H\left(\hat{f}_{kd}^D\right)\right].
\end{align}

With both the angle reciprocity and Doppler shift reciprocity, the unknown estimated coefficient of each scheduled UAV at the downlink channel training period is only the complex gain, which greatly decreases the training overhead. Besides, since the beamforming is executed at BS, there is no necessity for the scheduled UAVs to know the Doppler and angle signature of themselves. In this way, the feedback cost is greatly decreased for UAV communications. 

\subsection{Simplified DOA Tracking With Kalman Filter}
Denote the inverse discrete Foulier transformation (IDFT) of the channel $\mathbf{h}_k(l,(p+1)\eta)$ between BS and UAV $k$ in Equ.~\eqref{equ:SP4} as
\begin{align}\label{equ:DFT}
\tilde{\mathbf h}_k(l,(p+1)\eta)=\mathbf{F}^H\mathbf{h}_k(l,(p+1)\eta),
\end{align}
where $\mathbf F$ is the normalized $M^B\times M^B$ IDFT matrix with $\left[\mathbf F\right]_{rq}=e^{j\frac{2\pi}{M}rq}/\sqrt{M^B}$.
On the basis of \eqref{equ:DFT}, the $q$-th element of $\tilde{\mathbf h}_k(l,(p+1)\eta)$ can be derived as \eqref{equ:channeldft}, which is shown on the top of the next page.

\begin{figure*}[!t]
	\normalsize
	\setcounter{MYtempeqncnt}{\value{equation}}
	\setcounter{equation}{56}
\begin{align}\label{equ:channeldft}
\left[\tilde{\mathbf h}_k(l,(p+1)\eta) \right]_{q}&=
\frac{1}{\sqrt{M^B}}\sum_{m=0}^{M^B-1}\alpha_{k}e^{-j2\pi  f_{kd} lN_bT_s }e^{j 2\pi m\left[\frac{q}{M}-\frac{d \sin\theta_{k}}{\lambda}- \frac{p\eta d\sin \theta_k }{c}\right]}\notag\\
&=\frac{1}{\sqrt{M^B}}\alpha_{k}e^{(-j2\pi  f_{kd} lN_bT_s-j\frac{M^B-1}{2}\eta_{k})}\frac{\sin(\frac{M^B2\pi m\left[\frac{q}{M^B}-\frac{d \sin\theta_{k}}{\lambda_c}- \frac{ p\eta d\sin \theta_k }{c}\right]}{2})}{\sin(\frac{2\pi m\left[\frac{q}{M^B}-\frac{d \sin\theta_{k}}{\lambda_c}- \frac{ p\eta d\sin \theta_k }{c}\right]}{2})}.
\end{align}
	\setcounter{equation}{\value{MYtempeqncnt}}
	\addtocounter{equation}{1}
	\hrulefill
	\vspace*{4pt}
\end{figure*}

When $M^B$ is large, there always exists $q$ meeting $\frac{q}{M^B}-\frac{d \sin\theta_{k}}{\lambda_c}- \frac{ p\eta d\sin \theta_k }{c}=0$ for the given $p$ and $\theta_{k}$. In this case, all the channel power is concentrated on this point $p$, which is given by
\begin{align}\label{equ:qq}
q=\frac{d M^B\sin\theta_{k}}{\lambda_c}+\frac{M^B p\eta d\sin \theta_k }{c}.
\end{align}

Equ.~\eqref{equ:qq} also indicates that UAVs with different DOAs will exhibit different spatial distribution.
According to~\eqref{equ:channeldft}, we can obtain
\begin{align}\label{equ:qck}
q_k(l)=\frac{dM}{\lambda}\sin {{\theta}_k}(l)\left(1+\frac{n\eta}{f_c}\right)+u_k(l),
\end{align}
where $u_k(l)$ is the measurement noise with variance ${Q}_{u_k}$, and Equ.~\eqref{equ:qck} can be named as the \emph{measurement equation} for DOA tracking.

Denote $\mathbf{\Psi}_k(l)=[\theta_k(l),\dot{\theta }_k(l)]$ as the system states of DOA Tracking, where $\theta_k(l)$ and $\dot{\theta}_k(l)$  represent DOA and angular rate of user $k$ in block $l$ respectively. Then, the kinematic model can be applied to characterize the variation of DOA as \cite{trackingangle2}
\begin{align}\label{equ:systemequation}
\mathbf \Psi_k(l)&=\mathbf{\Phi}\mathbf{\Psi}_k(l-1)+\boldsymbol \omega_k(l)\notag\\
&=
\begin{bmatrix}
1 &NT_s \\
 0& 1
  \end{bmatrix}\mathbf{\Psi}_k(l-1)+\boldsymbol \omega_k(l),
\end{align}
where $\boldsymbol \omega_k(m)$ is the system noise that  meets $\E[\boldsymbol\omega_k(m)\boldsymbol\omega_k^H(m)]=\mathbf{Q}_{\mathbf \omega_k}$, and (\ref{equ:systemequation}) can be named as the \emph{system equation} for DOA tracking.

According to \eqref{equ:systemequation} and \eqref{equ:qck}, the DOA tracking procedure is a typical nonlinear system. In this case, extended Kalman filter~(EKF) would serve as a common approach for DOA tracking, and the detailed steps are illustrated in Alg.~\ref{alg:algorithm2}.
\begin{algorithm}[t]
\caption{: DOA Tracking Algorithm}
\label{alg:algorithm2}
\begin{itemize}
\item \textbf{Step 1:} {Initialization:} obtain the prior DOA information as
 $\hat{{ \boldsymbol \theta}}_k(1)=[ \theta_k(1),0]$, $\boldsymbol\kappa(1)=\mathbf 0$;
\item \textbf{Step 2:} {Compute the Jacobi matrix for the system equation:}
$\boldsymbol \phi_{k}(l-1)=\frac{\partial [\boldsymbol{\psi}\boldsymbol{\theta}_k(l-1)]}{\partial \boldsymbol{\theta}_k(l-1)}=\boldsymbol \phi.$
\item \textbf{Step 3:} {Compute the Jacobi matrix for the measurement equation:}
According to~\eqref{equ:channeldft}, we can derive the elements of the Jacobi matrix corresponding to the $p+1$-th sub-carrier as $\varpi_{kp}(l-1)=\frac{dM}{\lambda}\frac{\partial [\sin \bar\theta_k(l-1)\left(1+\frac{n\eta}{f_c}\right)]}{\partial \mathbf \theta_k(l-1)}=\frac{dM}{\lambda}\cos \bar\theta_k(l-1)\left(1+\frac{p\eta}{f_c}\right),$ and the Jacobi matrix can be stacked from $\varpi_{kp}(l-1)$ as $\boldsymbol{\varpi}_{k}(l-1)$.

\item \textbf{Step 4:} {Prediction of the system states:}
${{ \boldsymbol \theta}}_k(l|l-1)=\boldsymbol{\phi}{{ \boldsymbol \theta}}_k(l-1);$
\item \textbf{Step 5:} {Minimize the predicted mean square error~(MSE):}
$\boldsymbol\kappa(l|l-1)=\boldsymbol \phi_{k}(l-1)\boldsymbol\kappa(l-1)\boldsymbol \phi_{k}(l-1)^H+\mathbf{Q}_{\boldsymbol \omega_k};
$
\item \textbf{Step 6:} {Compute the Kalman gain matrix:}
$
\boldsymbol\Upsilon(\zeta)=\boldsymbol\kappa(l|l-1)\boldsymbol\varpi_{k}(l-1)^H[\boldsymbol\varpi_{k}(l-1)\boldsymbol\kappa(l|l-1)\boldsymbol\varpi_{k}(l-1)^H+{Q}_{ u_k}]^{-1};
$
\item \textbf{Step 7:} {DOA tracking:}
$
{{ \boldsymbol \theta}}_k(l)={{ \boldsymbol \theta}}_k(l|l-1)+\boldsymbol\Upsilon(l)[q_{c,k}(l)-\frac{dM\sin \bar\theta_k(l|l-1)}{\lambda}];
$
\item \textbf{Step 8:} {Compute minimum mean square error~(MMSE):}
$
\boldsymbol\kappa(l)=[\mathbf I-\boldsymbol\Upsilon(l)\boldsymbol \varpi_{k}(l)]\boldsymbol\kappa(l|l-1);
$
\item \textbf{Step 9:} Go to next block $l +1$.
\end{itemize}
\end{algorithm}

According to~Alg.~\ref{alg:algorithm2}, the DOA can be realtimely tracked by the Kalman filter based predicting and updating. Channel tracking is transmitted to tracking the Doppler information and complex gain information, which decreases the training overhead. The whole channel tracking procedure is concluded in Fig.~\ref{fig:system}. For clarity, we summarize  channel tracking procedure:
\begin{itemize}
	\item Uplink DOA, Doppler shift, and complex gain tracking;
	\item Uplink channel reconstruction;
	\item Downlink DOA and Doppler derivation with angle and Doppler shift reciprocity;
	\item Downlink complex gain tracking;
	\item Downlink channel reconstruction.
\end{itemize}

\begin{figure*}[t]
	\centering
	\includegraphics[width=130mm]{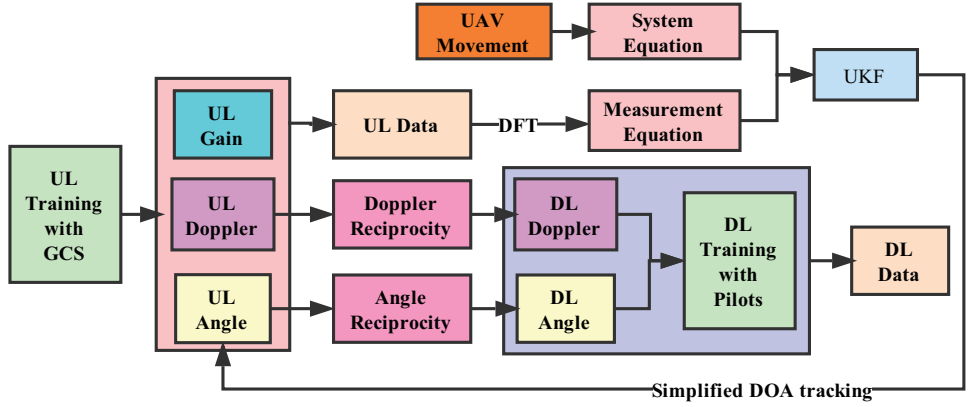}
	\caption{Efficient channel tracking strategy for UAV communications with mmWave massive array antenna.}
	\label{fig:system}
\end{figure*}

\section{Simulations}\label{sec:simulation}
In this section, various simulation results are provided to verify the effectiveness of the proposed method. The dimension of BS antenna is $M^B=128$, and the antenna spacing is set as the half of the carrier wavelength, namely, $d=\frac{\lambda}{2}$. The channel carrier frequency is set as $f_c=60$ GHz, and the bandwidth is set as $W = 600$ MHz. There are $K=4$ UAVs with single antenna that are uniformly distributed in the cell. The performance criteria are set as the normalized channel gain, Doppler shift, DOA, and uplink\ downlink channel, which is given by
\begin{align}\label{equ:SE}
&\textup{MSE}_{\mathbf h_k(l)}= \frac{1}{LK}\sum_{k=1}^K\sum_{l=1}^L\frac{\left\|\mathbf h_k(l)-\hat{\mathbf h}_k(l)\right\|^2}{\left\|\mathbf h_k(l)\right\|^2},\\
&\textup{MSE}_{\alpha_{k}}= \frac{1}{LK}\sum_{k=1}^K\sum_{l=1}^L\frac{\left\|\alpha_{k}-\hat\alpha_{k}\right\|^2}{\left\|\alpha_{k}\right\|^2},\\
&\textup{MSE}_{f_{kd}}= \frac{1}{LK}\sum_{k=1}^K\sum_{l=1}^L\frac{\left\|f_{kd}-\hat f_{kd}\right\|^2}{\left\|f_{kd}\right\|^2},\\
&\textup{MSE}_{\theta_{k}}= \frac{1}{LK}\sum_{k=1}^K\sum_{l=1}^L\frac{\left\|\theta_{k}-\hat\theta_{k}\right\|^2}{\left\|\theta_{k}\right\|^2}.
\end{align}


\begin{figure}[t]
	\centering
	\includegraphics[width=95mm]{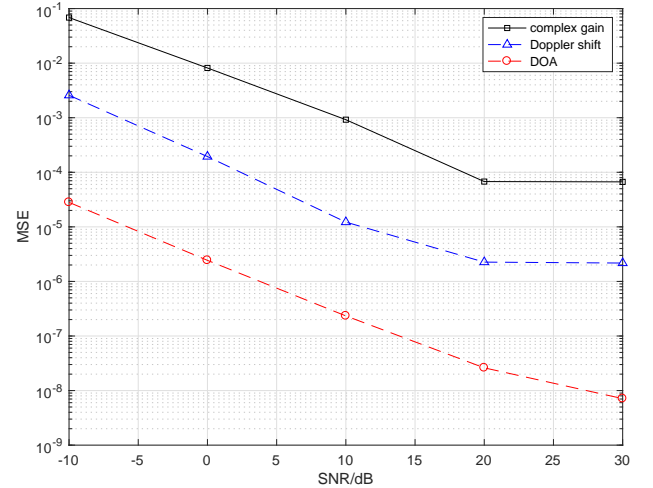}
	\caption{The performance of parameter estimation over SNR.}
	\label{fig:1}
\end{figure}

We first investigate the performance of the proposed GCS based parameter estimation method in Fig.~\ref{fig:1}. The performance metric is the normalized mean square error~(MSE). It can be found that with the increase of the signal to noise ratio~(SNR), MSEs of the complex gain, the Doppler shift, and DOA all decrease. Besides, we see that there are error floors of the estimated parameter due to the limited iteration step in Alg.~\ref{alg:algorithm1}. 

Next, we investigate the  performance of the proposed channel tracking method over SNR in Fig.~\ref{fig:doasnr}, where the antenna number is set as $M=16$, $32$, and $128$ respectively. It can be seen that the MSEs of the proposed GCS based channel tracking methods decreases with the increase of SNR. Besides, with the increase of the antennas, the performance of the proposed method is enhanced, since much more antennas would bring much more spatial gain. Moreover, the conventional channel tracking method that ignores the beam squint effect would not work for UAV communications with large antenna array, which verifies the effectiveness of the proposed method.

\begin{figure}[t]
	\centering
	\includegraphics[width=95mm]{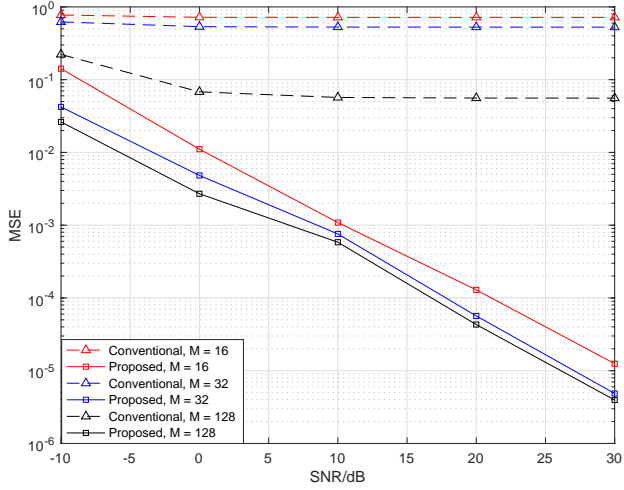}
	\caption{The performance of the proposed channel tracking method over SNR with different number of antennas.}
	\label{fig:doasnr}
\end{figure}

\begin{figure}[t]
	\centering
	\includegraphics[width=95mm]{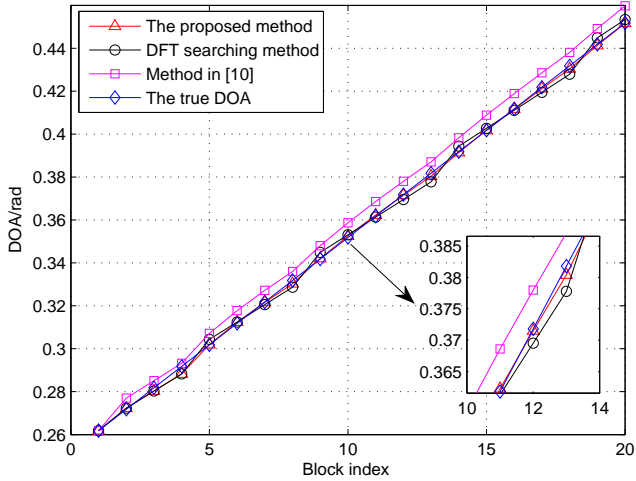}
	\caption{The performance of the DOA tracking over SNR with different number of antennas.}
	\label{fig:doatime}
\end{figure}


Then, we investigate the DOA tracking performance over the time in Fig.~\ref{fig:doatime}, where the DFT searching method, the method in~[10], and the true DOA are also displayed for comparison. It can be seen that the tendencies of all the displayed methods are consisted with the true DOA since the DFT of the channel could reflect the DOA distribution of the scheduled users. Besides, both the DFT searching method and the proposed method are superior to the method in [10], which neglects the beam squint effect. The performance of the proposed method is much better than the DFT searching method. The reason is that the proposed method and the DFT searching method all take the antenna selective effect into consideration, and the performance of the proposed method is further enhanced by the Kalman filter.

\section{Conclusions}\label{sec:conclusion}
In this paper, we investigated UAV communications with mmWave massive array antenna.  First, we explored the UAV channel under both Doppler shift and beam squint effect. Then, we proposed an efficient channel tracking method for mmWave UAV communication systems, where the channel could be derived by estimating DOA, Doppler shift, and complex gain information of the incident signal, respectively. The gridless compressed sensing method was exploited to track the channel parameters of UAV communications with massive antenna array. Finally, we provided various simulation results to verify the effectiveness of the proposed method over the existing ones.

\end{document}